\documentclass[aps]{revtex4}
\usepackage{graphicx}
\usepackage{amsthm}
\usepackage{amssymb, amsmath}
\usepackage{color}
\usepackage{soul,xcolor}

\begin{document}
\title{Recursive simulation of quantum annealing}

\author{A P Sowa$^{1}$, M J Everitt$^{2}$,  J H Samson$^{2}$, S.E. Savel'ev$^{2}$, A M Zagoskin$^{2}$, S Heidel$^1$, J C Z\'u\~niga-Anaya$^3$}
\affiliation{$^{1}$Department of Mathematics and Statistics, University of Saskatchewan}
\affiliation{$^{2}${Department of Physics, Loughborough University, Loughborough, Leics LE11 3TU, UK}}
\affiliation{$^3$ Research Computing Group, ICT, University of Saskatchewan }

\newtheorem{definition}{Definition}[section]
\newtheorem{theorem}{Theorem}[section]
\newtheorem{lemma}{Lemma}[section]
\newtheorem{corollary}{Corollary}[section]

\begin{abstract}
The evaluation of the performance of adiabatic annealers is hindered by lack of efficient algorithms for simulating their behaviour. We exploit the analyticity of the standard model for the adiabatic quantum process to develop an efficient recursive method for its numerical simulation in case of both unitary and non-unitary evolution.   Numerical simulations show distinctly different distributions for the  most important figure of merit of adiabatic quantum computing --- the success probability --- in these two cases.
\end{abstract}
\date{\today}
\maketitle
\section{Introduction}\label{Introduction}
 The ongoing debate whether the devices D-Wave One and D-Wave Two, which contain hundreds of superconducting flux qubits, demonstrate quantum or classical annealing \cite{Boxio,Ronnow,Bian,Smolin,Wang,Shih,Crowley,Amin}, is to a large extent due to the fact that their quantitative simulation is at the moment impossible. Characteristically, there was no controversy over quantum behaviour of an 8-qubit register \cite{Johnson}, where a quantitative agreement was found between the experimental data and the theoretical predictions. Though one cannot expect to be able to model classically a large enough quantum system \cite{Feynman}, the   development of efficient classical algorithms could hopefully enable us to model quantum systems just big enough to take over. We can also  exploit the idiosyncratic features of certain quantum  devices, which simplify their modelling by classical means. Here we demonstrate such an algorithm for modelling linear adiabatic evolution, including the case when decoherence is present. This result depends on the analyticity of the final state of the system as a function of evolution time, which is rigorously proven. Our small-scale numerical realizations of this algorithm show that some features of quantum annealing are comparable to the behaviour of D-Wave machines. Since at this stage the simulations do not account for the specific physical conditions that influence the machine's performance we stop short of drawing conclusions as to the strength of signatures of quantumness in the actual device. Nevertheless, we do hope that the tools developed here will enable a better understanding of this important problem.

An adiabatic process evolves the state vector of a quantum annealer over time $T$ along the trajectory $s\mapsto |\psi(s)\rangle$, where the parameter $s\in[0,1]$ is the reduced time $s=t/T$. The initial value $|\psi(0)\rangle$ is the ground state of the initial Hamiltonian $H_i$. When the process stops at $t=T$ the state $|\psi(1)\rangle$ is considered as an  approximant to the ground state of the custom design Hamiltonian $H_f$, which encodes the solution of an optimization problem of choice. We consider the standard linear case, when the evolution is driven by the time-dependent Hamiltonian
\begin{equation}\label{Ham_of_s}
H(s) = (1-s) H_i +s H_f,
\end{equation}
according to the Schr\"{o}dinger equation
\begin{equation}\label{Schr_of_s}
\frac{d}{ds} |\psi(s)\rangle= -i(T/\hbar)\, H(s)\,|\psi(s)\rangle.
\end{equation}

Let the quantum system be an $N$-qubit register.  In order to fix notation we describe the eigenstates and the corresponding eigenvalues of $H(s)$:
\begin{equation}
H(s) \,|m;s\rangle = E_m(s)\,|m;s\rangle,\quad  \mbox{ where }\,\, E_0(s)\leq E_1(s) \leq \ldots \leq E_{2^N-1}(s).
\end{equation}
However, this assumption is not necessary to derive the analytic results presented in next section. As we will see, our analysis remains valid in every case of bounded (in particular, finite) Hamiltonians $H_i, H_f$.

The purpose of analysis and simulation of (\ref{Ham_of_s}) is estimation of the success probability
\begin{equation}\label{success}
P = \|\Pi\psi(1)\|^2 = \sum_{j = 0,1,\ldots d-1} |\langle \psi (1)| j; 1\rangle|^2 ,
\end{equation}
where $\Pi$ denotes the projector onto the ground state space of $H_f$, e.g. for a non-degenerate ground state $P = |\langle \psi(1)| 0;1\rangle |^2$. (In fact, in most cases of interest to us the dimension of the ground state space $d$ exceeds $1$. This is rigorously accounted for in all numerical results presented in this article.) $P$ provides a measure of the accuracy with which the adiabatic process finds the solution of the underlying problem. For the Hamiltonians considered here, the instantaneous ground state $|0; s\rangle$ is non-degenerate for $0 \le s < 1$ (see Remarks), although the ground state of $H_f$ may be highly degenerate.  In a practical optimisation problem any of these degenerate states would be an equally useful solution. While there is a unique vector in the ground state space adiabatically connected to the initial ground state, this will in general be a superposition of computational basis states (a random one of which will be seen in the readout) and not necessarily of interest in the context of optimisation.

Having fixed $T$ and the number of steps and  the Hamiltonian, we repeatedly conduct the numerical experiment. Final Hamiltonians are drawn from an ensemble of binary optimization problems and the simulation outputs the probability of success $P$ for each instance of $H_{f}$. We record the outcome of the series of experiments in a histogram of $P$ values. The histogram of the probability of success has long been adopted as a means for testing hypotheses about adiabatic computing, see e.g. \cite{Farhi} and, more recently, \cite{Boxio}.

\vspace{.5cm}

\noindent \textbf{Remarks.} It is interesting to briefly consider the evolution of the ground state projector  $s\mapsto \Pi_s$, where $\Pi_1 = \Pi$ is as above.   Given that in the case at hand the dependence of $H(s)$ on $s$ is linear and \emph{a fortiori} holomorphic, the theory of analytic perturbations applies and it gives us the following picture: A standard general result, see Theorem 1.8, p. 70, \cite{Kato}, tells us that all the eigenvalues $E_k(s)$ as well as the corresponding eigenprojections, including $\Pi_s$, will depend on s analytically in the entire complex plane. \emph{A priori} these analytic functions may have algebraic type singularities at a set of isolated exceptional points in the complex plane, where the number of distinct eigenvalues is lower than in the generic case. More precisely, the eigenvalues as functions of $s$ are holomorphic in simply-connected domains that exclude the exceptional points. At the exceptional points an eigenvalue function may have, although does not necessarily have to have, a branching point. Moreover, at a branching point the corresponding eigenprojection will have a pole. \emph{A priori}, such a scenario, i.e. hitting a branching point where the eigenprojection has a singularity, might be construed as a “phase transition” of the first kind. From this point of view it is not clear  \emph{a priori} why a result such as Theorem \ref{Radius_infty} should hold. (Note that the result can be reformulated for complex variable $s$ stating that the solution of (\ref{Schr_of_s}) is an entire function of $s$.)

If we restrict the variable $s$ to the real line, the picture outlined above is considerably simplified. Indeed, in this case $H(s)$ is a self-adjoint operator, and satisfies the assumptions of Theorem 6.1, p. 120, \cite{Kato}. This result assures us that both the eigenvalues and the eigenprojections are holomorphic around every point of the real line. In other words, even if some $s_0$ on the real line is an exceptional point, the projection will not have a pole there and none of the eigenvalue functions will branch. However, we emphasize that the dimension of the ground state may jump discontinuously at the exceptional point, e.g. consider $E_0(s)$ with the corresponding projection $\Pi_s$ and $E_1(s)$ with $\tilde\Pi_s$. It may be that $E_0(s_0)=E_1(s_0)$. Both these projections depend on $s$ holomorphically but the projection onto the ground state of $H(s_0)$ will be $\Pi'_{s_0} = \tilde\Pi_{s_0} + \Pi_{s_0}$. This explains why the projections corresponding to the eigenvalues depend on $s$ continuously, while the projection to the ground state may have isolated jump discontinuities. In particular, the latter scenario is bound to play out for at least one point $s_0\in [0,1]$ when the ground state spaces of $H_i$ and $H_f$ have different dimensions. At this juncture it is interesting to invoke a recent result of Zhang et al., \cite{Zhang}, which shows that in the particular case of Hamiltonians we work with in our numerical simulations, see (\ref{Hi})--(\ref{Hf}), the ground state is non-degenerate for all $s\in[0,1)$. This implies that in all cases for which the ground state of $H_f$ is degenerate the minimal gap tends to zero as $s$ approaches $1$ from the left.

\section{Consequences of analyticity}
Since Hamiltonian (\ref{Ham_of_s}) depends on parameter $s$ linearly, the solution of Eq.~(\ref{Schr_of_s}) may be expected to depend on $s$ analytically, at least for small $s$. Thus, we look for solutions in the form of a Taylor series
\begin{equation}\label{Ansatz}
\psi(s) = \psi_0 +s\psi_1 +s^2\psi_2 + s^3\psi_3+\ldots
\end{equation}
In fact, we will demonstrate  that the radius of convergence of series characterizing solutions of Eq.~(\ref{Schr_of_s}) is infinite, i.e. the series converge for all $s$, see Theorem \ref{Radius_infty} below.
In order to simplify notation, let us define auxiliary operators
\begin{equation}\label{auxil}
A = -i(T/\hbar)\, H_i \quad \mbox{and }\,\, B=-i(T/\hbar)\,(H_f-H_i).
\end{equation}
Thus, (\ref{Schr_of_s}) is equivalent to
\begin{equation}\label{Schr_AB}
\frac{d}{ds} |\psi(s)\rangle= (A+Bs)\,|\psi(s)\rangle.
\end{equation}
Below we will consider this equation for general bounded operators $A$ and $B$, and only later draw conclusions about the specific case (\ref{auxil}) relevant to the application at hand.

Substituting (\ref{Ansatz}) in (\ref{Schr_AB}) we readily obtain
\begin{equation}\label{recurrence}
\begin{array}{lll}
|\psi_0\rangle &=& |\psi(0)\rangle \\
|\psi_1\rangle & =& A \,|\psi_0\rangle \\
\vdots & &\\
|\psi_n\rangle & = &\frac{1}{n} \left(A\, |\psi_{n-1}\rangle + B\, |\psi_{n-2}\rangle\right)\quad \mbox{ for } n\geq 2
\end{array}
\end{equation}
We use this recurrence as the core of numerical schemas for simulation of the solutions of (\ref{Ham_of_s}). It allows relatively efficient, as compared to ODE simulation, computation of consecutive terms and consecutive partial sums of the series (\ref{Ansatz}). Numerical experiment shows that this method, implemented directly as a MATLAB script, is typically faster than the standard ODE solver  `ode45' by a factor exceeding 3 for $N\leq 12$ qubits. For $N=14$, in which case the memory constraint plays a significant role, the speed up factor is approximately 60, reducing the computation time from just under one hour to under one minute. Moreover, as described in Section V, a careful implementation of an optimized code allowed us to simulate systems with up to $N=24$.

A similar procedure can be applied if the interpolation is a polynomial of order $m$ in $s$: in that case the three-term recurrence relation will be replaced by an $m+2$-order recurrence.  Indeed, this may also be possible for an analytic interpolating function of known Taylor series.  However, to demonstrate the method we only discuss the simplest (linear) interpolation here.

Note that evaluation of the success probability (\ref{success}) requires only the computation of
\[
P =  |\Pi\left(\psi_0 + \psi_1 + \psi_2 +\ldots\right) |^2.
\]
In other words, the solution is found, immediately as it were, at the point $s=1$ without the need of finding all the intermediate states $|\psi(s)\rangle$. This is in stark contrast to the computation based on an application of an ODE solver to Eq.~(\ref{Schr_of_s}). The current method has been tested against our earlier study of the correlation between gap and success probability\cite{Cullimore}; in that work we used a conventional (Dormand-Prince) solver.
 Note also that if series (\ref{Ansatz}) is known to converge absolutely, then $|\psi(s)\rangle$ automatically satisfies (\ref{Schr_of_s}). Therefore, the main issue at stake is the estimation of the radius of convergence
\begin{equation}\label{def_Rc}
R_{c} = \left(\limsup\limits_{n \rightarrow \infty} \|\psi_n\|^{1/n}\right)^{-1}.
\end{equation}
Here $\|\,\|$ denotes the $\ell_2$ norm. In general $R_c$ depends on the constituents of the process $H_i, H_f$ and possibly even $|\psi(0)\rangle$. From the point of view of simulations it is ideal to have $R_c = \infty$, which ensures that $|\psi(s)\rangle$ given in (\ref{Ansatz}) is the solution of (\ref{Schr_of_s}) for all (reduced) times $s = t/T$. If on the other hand, $R_c <1$, then the series  cannot be used to estimate the success probability. But, as we shall see, the series (\ref{Ansatz}) converges everywhere.

Let us introduce the following notation:
\begin{equation}
a := \|A\| , \quad b := \|B\| .
\end{equation}
Here, $\|\,\|$ is the operator norm.
Throughout the article we assume $0< a,b < \infty$. For the particular case of (\ref{auxil}), we have $a  = (T/\hbar)\,\|H_i\|, \quad b =  (T/\hbar)\,\|H_i-H_f\|$.  	The Hamiltonians are nontrivial and bounded, e.g. finite dimensional, and, moreover, $H_i\neq H_f$. We have the following result:

\begin{theorem}\label{Radius_infty}
$R_c = \infty$, i.e. series (\ref{Ansatz}) converges absolutely in the entire complex plane and its limit $\psi(s)$ for $s\in [0, \infty)$ is the solution of (\ref{Schr_AB}) satisfying the initial condition $\psi(0) = \psi_0$.
\end{theorem}

\begin{proof}
Recurrence (\ref{recurrence}) readily implies
\begin{equation}\label{psi_from_P}
|\psi_n\rangle = \frac{1}{n !}P_n\, |\psi_0\rangle,
\end{equation}
where $P_n = P_n(A,B)$ are operators defined via recurrence
\begin{equation}\label{recurrence_P}
\begin{array}{lll}
P_0 &= &I \\
P_1  &=& A \\
\vdots & &\\
P_{n+1} & = & AP_n + nBP_{n-1}\quad \mbox{ for } n\geq 1.
\end{array}
\end{equation}
Note that we have $\|P_0\| =1$, $\|P_1\| = a$, and $\|P_{n+1}\| = \| AP_n + nBP_{n-1}\| \leq a \|P_n\| +n b \|P_{n-1}\|$. We wish to estimate the rate of growth of $\|P_n\|$. To this end let us consider an auxiliary scalar sequence $(p_n)$ defined via recurrence
\begin{equation}\label{recurrence_p}
\begin{array}{lll}
p_0 &= &1 \\
p_1  &=& a \\
\vdots & &\\
p_{n+1} & = & ap_n + nbp_{n-1}\quad \mbox{ for } n\geq 1.
\end{array}
\end{equation}
Since $p_0 = \|P_0\|, p_1 =\|P_1\|$ and $p_n$ grows at least as fast as $\|P_n\|$, we clearly have
\begin{equation}\label{stage}
\|P_n\| \leq p_n.
\end{equation}
 Next, we undertake to estimate the growth rate of $(p_n)$. First, observe that $p_n = p_n(a,b)$ may be viewed as a polynomial in two variables, e.g. $p_2 = a^2 +b$, $p_3 = a^3 + 3ab$, etc. It is easily seen that the general form of the polynomials is
 \begin{equation}\label{gen_p}
 p_n = c_0[n]a^n + c_1[n] a^{n-2}b + c_2[n] a^{n-4}b^2 +  c_3[n] a^{n-6}b^3 + \ldots  = \sum\limits_{k=0}^{[\frac{n}{2}]}c_k[n]\, a^{n-2k}b^k,
 \end{equation}
 where the coefficients $c_k[n]$ remain to be found. Note that the last term of the polynomial is either $c_{n/2}[n]\, b^{n/2}$ (when $n$ is even) or $c_{[n/2]}[n]\, a b^{[n/2]}$ (when $n$ is odd), with $[x]$ denoting the integer part of $x$.
Moreover, it is easily seen that (\ref{recurrence_p}) implies
\begin{equation}\label{recurrence_pp}
\begin{array}{lll}
c_0[n] &= &1  \\
&& \\
c_1[n]  &=& \binom{n}{2} \\
\vdots & &\\
c_k[n]& = & c_k[n-1] + (n-1) c_{k-1}[n-2]
\end{array}
\end{equation}
Using this, and applying induction, one readily obtains an explicit formula
\begin{equation}\label{c_formula}
c_k[n] = (2k-1)!!  \binom{n}{2k},
\end{equation}
where $(2k-1)!! = 1\cdot3\cdot 5\cdot \ldots (2k-1)$.
In light of this (\ref{gen_p}) yields
\begin{equation}\label{term_T}
\frac{1}{n!} \, p_n = \sum\limits_{k=0}^{[\frac{n}{2}]}\frac{1}{k!(n-2k)!} a^{n-2k}\frac{b^k}{2^k} .
\end{equation}
  Next, we make the following observation
 \begin{equation}\label{inequality}
 k! (n-2k)! \geq \left[\frac{n}{3}\right]!\quad \mbox{ for } k = 0,1,2,\ldots , [n/2] .
 \end{equation}
 Indeed, we either have $k\geq \left[\frac{n}{3}\right]$ or else $n-2k > n -2 \left[\frac{n}{3}\right] \geq \left[\frac{n}{3}\right] $. In either case (\ref{inequality}) follows trivially.

 Next,  we obtain from (\ref{term_T}) and (\ref{inequality}):
 \begin{equation}\label{term_T_est}
 \begin{array}{lll}
 \frac{1}{n!} \, p_n &=& \sum\limits_{k=0}^{[\frac{n}{2}]}\frac{1}{k!(n-2k)!} a^{n-2k}\frac{b^k}{2^k} \\
 & \leq & \frac{1}{\left[\frac{n}{3}\right]!}a^n\{1 + \frac{b}{2a^2} + \left(\frac{b}{2a^2}\right)^2  + \left(\frac{b}{2a^2}\right)^3 +\ldots + \left(\frac{b}{2a^2}\right)^{[n/2]}\} \\
 & \leq & \frac{1}{\left[\frac{n}{3}\right]!} a^n (1+ \frac{b}{2a^2})^{n/2}.
 \end{array}
\end{equation}
 As is well known, $\left(\left[\frac{n}{3}\right]!\right) ^{1/n} \rightarrow \infty$ as $n\rightarrow \infty$. Thus, recalling definition (\ref{psi_from_P}), we obtain
 \begin{equation}\label{psi_est}
 \|\psi_n\|^{1/n} \leq \|\frac{1}{n!}P_n\|^{1/n} \|\psi_0\|^{1/n}\leq \left(\frac{1}{n!}\,p_n\right)^{1/n} \|\psi_0\|^{1/n} =
 \frac{1}{\left(\left[\frac{n}{3}\right]!\right)^{1/n}} a\left(1+ \frac{b}{2a^2}\right)^{1/2} \|\psi_0\|^{1/n}\rightarrow 0,
 \end{equation}
 which by (\ref{def_Rc}) implies $R_c = \infty$. $\Box$
\end{proof}

\noindent
\textbf{Remark 1.} Note that when $b < 2a^2$  (valid in the adiabatic regime)  estimate (\ref{term_T_est}) may be replaced by a more efficient one. Indeed in such a case
\[
1 + \frac{b}{2a^2} + \left(\frac{b}{2a^2}\right)^2  + \left(\frac{b}{2a^2}\right)^3 +\ldots  \leq \frac{1}{1-\frac{b}{2a^2}}
\]
If in addition $a<1$, then terms $p_n/n!$ diminish very fast which ensures fast convergence of series (\ref{Ansatz}) when $s=1$, also in the numerical sense.  However, this is not of interest in the current problem as it represents the anti-adiabatic limit, as the number of spin flips during the evolution is of order $a$. In the special case of interest $a  = (T/\hbar)\|H_i\|, b =  (T/\hbar)\|H_i-H_f\|$ and the condition $a<1\,\& \, b< 2a^2$ is equivalent to
\[
\hbar \,\frac{\|H_i-H_f\|}{2\,\|H_i\|^2} < T < \hbar\, \frac{1}{\|H_i\|},\, \mbox{ which implies }
\|H_i-H_f\| < 2\,\|H_i\|.
\]
This assumption on $T$ ensures the most efficient computation of series (\ref{Ansatz}). However, we do not claim that this restriction delineates the only regimes in which computation is effective.

\section{The master equation} \label{master_subsection}

Now we will extend the above results to the evolution described by the master equation (see e.g. \cite{Percival}):
\[
\frac{d}{dt}\,\rho = - \frac{i}{\hbar} [H,\rho] + L\rho L^{\dag} - \frac{1}{2} L^{\dag}L\rho - \frac{1}{2} \rho L^{\dag}L .
\]
Here, $\rho$ is the density matrix, $H$ is the quantum system Hamiltonian and the Lindblad operator $L$ accounts for decoherence due to the uncontrollable interaction of the quantum system with the environment. In some applications the master equation involves a finite number of different Lindblad operators. While for our purposes it is sufficient to have just one Lindblad, the results presented here can easily be extended to the more general case. From now on we will only consider the evolution of finite-dimensional systems. Furthermore, as before we assume that $H$ takes the special form (\ref{Ham_of_s}) with $s = t/T$. Substituting variables (\ref{auxil}) as before, we obtain an equivalent form of the master equation
 \begin{equation}\label{Lind}
 \frac{d}{ds}\,\rho =  \mathcal{L}_A [\rho] + s \mathcal{L}_B [\rho]
 \end{equation}
 with the shorthand notation
 \begin{equation}\label{def_LA}
 \mathcal{L}_A [\rho]: = [A, \rho] + T\left(L\rho L^{\dag} - \frac{1}{2} L^{\dag}L\rho - \frac{1}{2} \rho L^{\dag}L \right), \quad \mbox{ and } \quad
  \mathcal{L}_B [\rho] := [B, \rho].
 \end{equation}
Note that the evolution equation (\ref{Lind}) generalizes the von Neumann equation
 \begin{equation}\label{Schr_rho}
\frac{d}{ds} \rho(s)= -i(T/\hbar)\, \, [H(s), \rho(s) ].
 \end{equation}
In the context of adiabatic computing, it is vital to realize that there is an essential difference between the two types of evolution. Namely, at least if certain general conditions on $H_f$ are satisfied, the adiabatic theorem applies to the case (\ref{Schr_rho}) and ensures that the probability of success increases with increasing $T$ and will tend to one as $T$ tends to infinity. In contrast, there is no known generalization of the adiabatic theorem, extending it to the case of (\ref{Lind}). In fact, in the course of numerical experimentations we have observed that for a relatively large $L$ the probability of success may decrease with increasing time $T$. Figure 3 illustrates this phenomenon. This observation agrees with the conclusions of Ref.~\cite{Sarandy} that there is a finite time window for the operation of a quantum annealer interacting with the thermal bath.

In order to describe the analytic properties of (\ref{Lind}) we introduce the Hilbert space structure on the linear space of $K\times K$ matrices (e.g. $K = 2^N$ for an $N$ qubit system). For $X\in \mathbb{C}^{K\times K}$ we define the Hilbert-Schmidt norm $\|X \|_{HS} = \mbox{Tr}\, (XX^\dag)^{1/2}$.   This norm endows $\mathbb{C}^{K\times K}$ with the Hilbert space structure.  We denote the corresponding operator norms of $\mathcal{L}_A, \mathcal{L}_B:  \mathbb{C}^{K\times K} \rightarrow \mathbb{C}^{K\times K}$ simply by $\|\mathcal{L}_A \|, \|\mathcal{L}_B \|$. As is well known, the Hilbert-Schmidt norm has the submultiplicative property $\|A B\|_{HS}\leq \|A \|_{HS} \,\|B \|_{HS}$ which, when applied to (\ref{def_LA}), yields the following estimates:
 \begin{equation}\label{norm_est}
\|\mathcal{L}_A \| \leq 2 \|A\|_{HS} + 2 T\, \|L\|_{HS} ^2, \quad \|\mathcal{L}_B \| \leq 2 \|B\|_{HS} .
 \end{equation}

 Next, we wish to consider solutions of (\ref{Lind}). To this end we introduce the Taylor series Ansatz
\begin{equation}\label{Ansatz_rho}
\rho(s) = \rho_0 +s\rho_1 +s^2\rho_2 + s^3\rho_3+\ldots
\end{equation}
As before, we substitute (\ref{Ansatz_rho}) into (\ref{Lind}). In this way we obtain the following description of the Taylor series coefficients:
\begin{equation} \label{rho_n}
 \rho_n = \frac{1}{n!}\, Q_n(\mathcal{L}_A , \mathcal{L}_B ) \rho_0,
\end{equation}
where $ Q_n = Q_n(\mathcal{L}_A , \mathcal{L}_B )$ is a polynomial in variables $\mathcal{L}_A$ and $\mathcal{L}_B$ defined by recursion:
\begin{equation}\label{recurrence_Q}
\begin{array}{lll}
Q_0 &= &I \\
Q_1  &=& \mathcal{L}_A  \\
\vdots & &\\
Q_{n+1} & = & \mathcal{L}_A Q_n + n \mathcal{L}_B Q_{n-1}\quad \mbox{ for } n\geq 1.
\end{array}
\end{equation}
We observe that the proof of Theorem \ref{Radius_infty} can be repeated almost verbatim with obvious modifications such as: replacing  $A$ with $\mathcal{L}_A$,  $B$ with $\mathcal{L}_B$, and $P_n$ with $Q_n$, and redefining  $a:=\|\mathcal{L}_A\|$ and $b:=\|\mathcal{L}_B\|$. This yields the following result:

\begin{theorem}\label{Radius_infty_rho}
Series (\ref{Ansatz_rho}) with coefficients defined via (\ref{rho_n}) and (\ref{recurrence_Q}) converges absolutely in the entire complex plane. Moreover, its limit $\rho(s)$ for $s\in [0, \infty)$ is the solution of (\ref{Lind}) satisfying the initial condition $\rho(0) = \rho_0$.
\end{theorem}

\noindent
\textbf{Remark 2.} Consider the unitary map $U(s)$ defined by the adiabatic Schr\"{o}dinger equation:
\[
\frac{d}{ds} U(s) = -i(T/\hbar)\, H(s)\, U(s), \quad \mbox{ so that } U(s)|\psi (0)\rangle = |\psi (s)\rangle.
\]
Note that recurrence (\ref{recurrence_P}) applied in the special case --- i.e. $A = -i(T/\hbar)\,H_i$, $B=-i(T/\hbar)\,(H_f-H_i)$ --- provides a constructive description or simulation of $U(s)$ via the formula
\[
U(s) = \sum\limits_{n=0}^\infty \frac{1}{n!}\, P_n.
\]
In particular this description of the semigroup $s\mapsto U(s)$ may be used to simulate the evolution of mixed states (\ref{Schr_rho}). Indeed,
 \[
\frac{d}{ds} \rho(s)= -i(T/\hbar)\, \, [H(s), \rho(s) ]\quad \Longrightarrow\quad \rho(s) = U(s)\rho(0) U(s)^*.
 \]

\section{Algorithm analysis}\label{alg_efficiency}

We will utilize elements of the proof of Theorem \ref{Radius_infty} to analyze the efficiency of the recurrence algorithm (\ref{recurrence}) with the following specific problem constituents:
\begin{eqnarray}
H_i &=& - \sum\limits_k \sigma_k^x, \quad k\in {1, 2, \ldots, N}, \label{Hi} \\
H_f &=& H_{\mbox{Ising}} = -\sum\limits_{k<l} J_{kl}\,\sigma_k^z\sigma_l^z, \quad k,l\in {1, 2, \ldots, N}.\label{Hf}
\end{eqnarray}
 Here, $J_{ij}\in\{-1, 1\}$ denote  random variables with uniform distribution $P(J_{ij}=1)=P(J_{ij}=-1)= 1/2$, and $N$ denotes the number of qubits.

First, in order to estimate the efficiency of a recurrence one needs to establish a stopping condition. It is desirable to stop recurrence when the approximation error is sufficiently small. However, little is known a priori about the solution that is sought in (\ref{recurrence}) and so the calculus methods do not seem to provide an easy estimate for the dependence of error on the degree of approximation $n$. As a substitute we propose the following stopping condition:
\begin{equation}\label{stopping}
\mbox{Compute } \Psi_k \mbox{ for } k = 1,2,\ldots n, \mbox{ where  } n \mbox{ is the smallest index such that } \|\Psi_n\|^{1/n} \leq \varepsilon .
\end{equation}
 Here, $\varepsilon$ is regarded as a preset measure of accuracy. In light of the proof of Theorem \ref{Radius_infty} we know that $n$ defined by (\ref{stopping}) is bound to be finite. Moreover, since the solution $\Psi(s)$ is an analytic function in the complex plane this is expected to provide a good estimate of the remainder error when $s$ is bounded, e.g. $s \sim 1$.

Next, consider the particular case $T = 1$ (recall also $\hbar =1$). Let us assume that the Taylor series is centred at $s_0$. In such a case $a = \|H_i + s_0(H_i-H_f)\| = O(\|H_f\|) = O(N^2)$, $b = \|H_i - H_f\| = O(\|H_f\|) = O(N^2)$ which readily implies
\begin{equation}\label{anb}
a\left(1+ \frac{b}{2a^2}\right)^{1/2} = \left(a^2+ \frac{1}{2}b\right)^{1/2} = O(N^2).
\end{equation}
It now follows from (\ref{psi_est}), (\ref{stopping}) and the Stirling formula that
\[
 a\left(1+ \frac{b}{2a^2}\right)^{1/2}  \sim \, \varepsilon \left(\left[\frac{n}{3}\right]!\right)^{1/n} \sim \, \varepsilon \, n^{1/3}, \quad \mbox{ i.e. }\,\,
 n \sim \frac{1}{\varepsilon^3}O(N^6)
\]

Now, let us consider $T>>1$. While Theorems \ref{Radius_infty} and \ref{Radius_infty_rho} ensure that the Taylor series defining solutions of Eqns. (\ref{Schr_AB}) and (\ref{Lind}) converge in the entire complex plane, due to the roundoff errors, the Taylor series need not converge for relatively large values of $T$ as intermediate sums become exponentially large in the adiabatic regime.  The largest term $x^{n}/n!$ in the Taylor expansion of $e^{x}$ occurs when $n\approx |x|$ and is of the order of $e^{|x|}$.   To obtain solutions in such cases we use the following observation. Defining $\tau = s-s_0$ we observe that $\psi(\tau)$ satisfies (\ref{Schr_AB}) where $s$ has been replaced by $\tau$  and $A$ replaced by $A +s_0B$ with the initial condition $\psi|_{\tau = 0} = \psi(s_0)$. Similarly, $\rho(\tau)$ satisfies (\ref{Lind}) where $s$ has been replaced by $\tau$  and $A$ replaced by $A +s_0B$ with the initial condition $\rho|_{\tau = 0} = \rho(s_0)$. This enables one to conduct computation in stages by partitioning the $s$-interval $[0,1]$ into smaller segments or subintervals. Practice indicates that
the best course of action is to carry out the computation in $\sim T$ subintervals. Since computation within each segment is based on a faster converging Taylor series we obtain convergence of the solution in the entire $s$-interval for larger values of $T$.

Note that the estimate (\ref{anb}) holds for every subinterval, which ensures that every segment requires the computation of $O(N^6)$ coefficients.
Finally, it is easily seen that the number of arithmetical operations necessary to find $n$ Taylor coefficients is of rank $O(n 2^N)$. Thus, the overall number of arithmetical operations required to carry out the computation is of rank $~O(T\cdot n \cdot 2^N) = O(T\cdot N^6\cdot 2^N) = O(2^N)$. This is a polynomial time algorithm when measured against the classical problem size $2^N$. Naturally, and not very surprisingly, it is exponential in the number of qubits. The important feature of this algorithm is that it scales linearly with time $T$.

As an example, we consider the Landau-Zener Hamiltonian $H(s) = (1-2s)\sigma_z + \Delta \sigma_x$. It is easily seen that its minimal gap is given by $E_1(s)-E_0(s) = 2\sqrt{\Delta^2 +(1-2s)^2}$, with the minimum occurring at $s=.5$. Consider the solution $\Psi(s)$ of (\ref{Schr_of_s}) with $\Psi(0)$ prescribed as the ground state of $H(0)$. For the sake of an example let us attempt to compute $\Psi(1)$ for $\Delta = 1, T/\hbar=20$. In this instance the Taylor polynomial centered at $s=0$ will not give meaningful results: experiment shows that the algorithm outputs $\Psi(1)= [-0.0455-0.1663i, -0.4112 + 1.1521i]\cdot 10^{11}$ after $100$ iterations. Moreover one does not obtain correct convergence  by simply increasing the number of iterations. Indeed, the norm of the simulated $\Psi(1)$ will converge to the value of $\sim 1.05$ instead of the correct $\|\Psi(1)\|=1$. This is due to rounding error in the large intermediate terms. Indeed, all computations are done in floating point arithmetic with precision in the order of $10^{-16}$. This means that the difference between any two consecutive real numbers in the order of $10^{11}$ is as large as $10^{-5}$, which seriously impedes computational precision. In such a case a split of the interval into two parts, $0\le s \le 0.5$ and $0.5 \le s \le 1$, results in smaller intermediate values and ensures convergence to the correct limit.  This may be verified by a nearly exact calculation obtained via an application of, say, $200$, instalments, see Fig. 11. The figure also shows the expected nonmonotonic convergence to unity with increasing $T$, with oscillations determined by the time-averaged Larmor frequency.”
In addition, we note that a symbolic implementation of our method may be used to reproduce the Taylor series of the special function solution. However, numerical evaluation of special functions is frequently quite problematic, and so the calculation via instalments offers a more practicable approach even if the special function solutions were known in closed-form.

Finally, we have used this specific example to benchmark the accuracy of our numerical method (termed TPT for Target-Point-Taylor) by comparing it with the standard fourth-order Runge-Kutta method (RK or `ode45'), \cite{MC}. Based on a simple test we have observed the following: The TPT method gives the value of $$\Psi(1) \simeq [ 0.509629891598850 + 0.766898007985489\,i, -0.226356412675608 - 0.317659555887512\,i],$$ and $$ P \simeq 0.999801214304354$$ with just two steps computing 350 Taylor coefficients at each step. Re-computation based on $200$ intermediate steps shows that the value of $\Psi(1)$ is accurate to $10$ significant digits while the value of $P$ is accurate to $9$ significant digits. Secondly, we have attempted to reach maximal accuracy via the RK method. In order to obtain comparable accuracy ($7$ significant digits for $\Psi(1)$ and $9$ significant digits for $P$), namely
$$\Psi(1)\simeq [0.509629851149443 + 0.766898035396728\,i,	-0.226356395949650	- 0.317659566414614\,i],$$ and $$ P \simeq 0.999801214234416$$
we were forced to use a very fine time-discretization $ds = 1.86\times 10^{-9}$ as larger time-steps lead to greater error. This example suggests that for comparable accuracy the ratio of the number of steps required by RK to those required by TPT exceeds  $5.4\times 10^8 : 700$. As mentioned above, the ODE-solver based methods are impracticable for larger systems due to long computation time.

\vspace*{.2cm}

\noindent
\textbf{Remark 3.} While the formal analysis above suggests that the constant of proportionality between the number of Taylor coefficients $n$ and $N^6$ may be very large ($\sim \varepsilon^{-3}$), numerical evidence suggests that a rigid \emph{a priori} number, say $n = 50...90$ is typically sufficient to ensure numerical convergence. However, for a small fraction of a percentage of Ising Hamiltonians there is no convergence under these conditions. Nevertheless, in the type of experiment considered in this article those cases are statistically insignificant and may be discarded.
\vspace*{.2cm}

\noindent
\textbf{Remark 4.} We note that the recurrence algorithm does not require that all the Taylor series coefficient terms be kept in the memory. Indeed, in the case of unitary evolution with the dynamic variable being a vector $\Psi$ of length $2^N$, the memory requirement is of rank $3\cdot 2^N$, regardless of the number of coefficients $n$ that are kept in the approximation. Indeed, we only need to keep in the memory the variable $\Psi$, which is continually updated, and the two last Taylor coefficients, $\Psi_k$ and $\Psi_{k-1}$ as $k$ runs from $2$ to $n$. For a similar reason the memory requirement in the non-unitary case, when the dynamic variable is $\rho$ of size $2^{2N}$, is of rank $3\cdot 2^{2N}$.
\vspace*{.2cm}

 \section{Numerical experiments}\label{section_numer}
We have conducted numerical experiments for both a pure adiabatic process (\ref{Schr_of_s}) and the process with dissipation (\ref{Lind}) using matrix Hamiltonian as in (\ref{Hi}) and (\ref{Hf}). The choice of values for $J_{ij}$'s is generated via MATLAB's generic random number generators. It is done independently for every run of the experiments discussed below.  Note that $\|H_{i}\|=N$ and $\|H_{f}\|\leq N(N-1)/2$.

Following the discussion on Section IV, we divide the s-interval in $\sim T$ subintervals. Even though convergence in each subinterval is faster, the computational time is generally increased by a constant factor not exceeding the number of stages.

The numerical simulations were implemented using compiled MATLAB code. The running time was lessened by using sparse matrices and refactoring the formulas (\ref{auxil}) and (\ref{recurrence}) to reduce the number of matrix-vector calculations in favour of vector-vector calculations.
The construction of $H_i$ and $H_f$ was made much more efficient by exploiting the structure and symmetries of the matrices. Memory requirements were also lessened by taking advantage of the sparseness and symmetry in the matrix computations.
As a result, $N$ up to $24$ becomes computable on a modest distributed-computing server.

Appendix II provides the source code of our current MATLAB implementation of the algorithm. The \texttt{Main} function executes the simulation the indicated number of times, using a sparse matrix \texttt{H\_i} (the Heisenberg matrix representation of the initial Hamiltonian) and a different sparse vector
\texttt{H\_f} (the diagonal of the Random Ising Hamiltonian in the Heisenberg matrix representation) at each time.
Functions \texttt{Initial\_Hamiltonian} and \texttt{Random\_Ising} create those sparse structures using fractal-like iterations. Finally, function \texttt{Taylor\_Installments} calculates the Taylor series.

The histograms presented in Appendix I were produced using up to $180$ cores of 2.0GHz Intel E5-2640L Xeon processors distributed on the compute cluster PLATO of the University of Saskatchewan.
To illustrate the performance of the implemented algorithm, a total of $21780$ individual simulations were ran for $N=20$ in about 14 hours, and using less than 5GB of memory per simulation.

Figures 9 and 10 below
illustrate the evolution of the compute time per simulation as a function of the number of qubits. An exponential increase of the compute time for large numbers of qubits can be observed. We expect this to be the case for any number of qubits and, in consequence,
compute times for simulations of larger systems can be estimated.
The fact that this exponential tendency is not verified when the number of qubits is too small, is certainly due to the fact that MATLAB, and the server itself, requires some time for internal processes (other than the execution of the algorithm). It appears that this time is much larger than the execution time of the algorithm for these cases when the number of qubits is too small.

The crucial limiting factor in simulation of larger systems is memory capacity. However, early stage evidence shows that further algorithm optimizations, specialized to 3TB RAM server architectures, enables simulations of size up to 36 qubits. A detailed account of this ongoing work goes outside the scope of this article and will be reported elsewhere.

\subsection{Unitary adiabatic evolution: Schr\"{o}dinger's equation} \label{experiment_unitary}

We use a schema based on recurrence (\ref{recurrence}). Furthermore, we set $|\psi_0\rangle$ to be the ground state of $H_i$.    We note that $H_{f}$ is exponentially degenerate: the Hilbert space is $2^{N}$-dimensional but there are only $O(N)$ distinct eigenvalues, which are separated by multiples of $2$.  In particular, the Hamiltonian $H(s)$ is invariant under a global spin flip for all time, so adiabatic evolution occurs in the subspace spanned by equal superpositions of computational states and the corresponding spin-flipped states, $\left\{\left(|abc\cdots\rangle+|\bar a \bar b \bar c\cdots\rangle\right)/\sqrt2\right\}$.  Let $\Pi$ denote the projector onto the ground state space of $H_f$. (Note that $\Pi$ is computed numerically, which presents no difficulties.) The probability of success is defined as $P = \|\Pi\psi(1)\|^2$.

 The results from simulation are presented via probability of success histograms, Fig. 1, and Figs. 4--8. Of note is the ragged landscape feature of the histograms.  This can be understood qualitatively in the following way. The success probability is largely determined by the low-lying eigenvalue structure. The aforementioned degeneracy reduces the number of qualitatively different classes and we conjecture that the peaks correspond to different classes. We conjecture a connection between the low-P tail of the distribution of probability of success and the existence of optimization problems that are inherently hard.  Since the optimization problem at hand is known to be NP hard, \cite{Barahona}, we expect this phenomenon to persist for larger systems.

\subsection{Non-unitary adiabatic evolution: Master equation with one Lindblad operator}

In this setting the probability of success is defined as $P = \mbox{Tr}\,(\Pi \rho(1))$ where, again, $\Pi$ is the projector onto the ground state space of $H_f$. In order to compute $\rho(1)$ we use a schema  based on recurrence (in notation of Subsection \ref{master_subsection}):
\begin{equation}\label{recurrence_Q_rho}
\begin{array}{lll}
\rho_0 &=& |\psi_0\rangle\langle \psi_0|,\quad \mbox{ where $\psi_0$ denotes the ground state of } H_i \\
\rho_1  &=& \mathcal{L}_A \rho_0 \\
\vdots & &\\
\rho_{n+1} & = & \frac{1}{n+1}\mathcal{L}_A\left( \rho_n + \mathcal{L}_B \rho_{n-1}\right)\quad \mbox{ for } n\geq 1.
\end{array}
\end{equation}
 For the particular experiments presented in Figure 2 we took $L = 0.1\hat{a}$ where $\hat{a}$ is the  operator, which in the eigenbasis of $H_f$ assumes the off-diagonal form
 \begin{equation}\label{ahat}
\left[ \begin{array}{cccccc}
 0 &1 & 0& 0 & 0 & \ldots \\
 0 & 0 & \sqrt{2} & 0 & 0 & \ldots \\
 0 & 0 & 0 & \sqrt{3} & 0 & \ldots \\
 \vdots & \vdots & \vdots & \vdots & \vdots & \vdots
 \end{array}\right]
 \end{equation}
The basis states are in order of increasing energy, with basis states in the degenerate subspaces listed in a particular order.  The effect of this $H_{f}$-dependent Lindblad operator, in isolation, would be to relax the system to the selected ground state, thereby breaking the degeneracy.  However,  since it has a highly non-local action on the qubits, its interaction with $H(s)$ disturbs the adiabatic evolution and destroys the fine structure of the probability histogram, see Fig. 3.
\vspace{.1cm}

\noindent
\textbf{Acknowledgements.}
APS thanks the Department of Physics, Loughborough University, for hosting him in the autumn of 2013 and acknowledges partial support of the Canadian Foundation for Innovation, grant LOF \# 22117. The authors are grateful to the referees for pointing out ways to improve the presentation of this material.

\newpage

\section*{Appendix I: Figures}

\begin{minipage}[t]{0.45\textwidth}
\vspace{1cm}
\includegraphics[width=\textwidth]{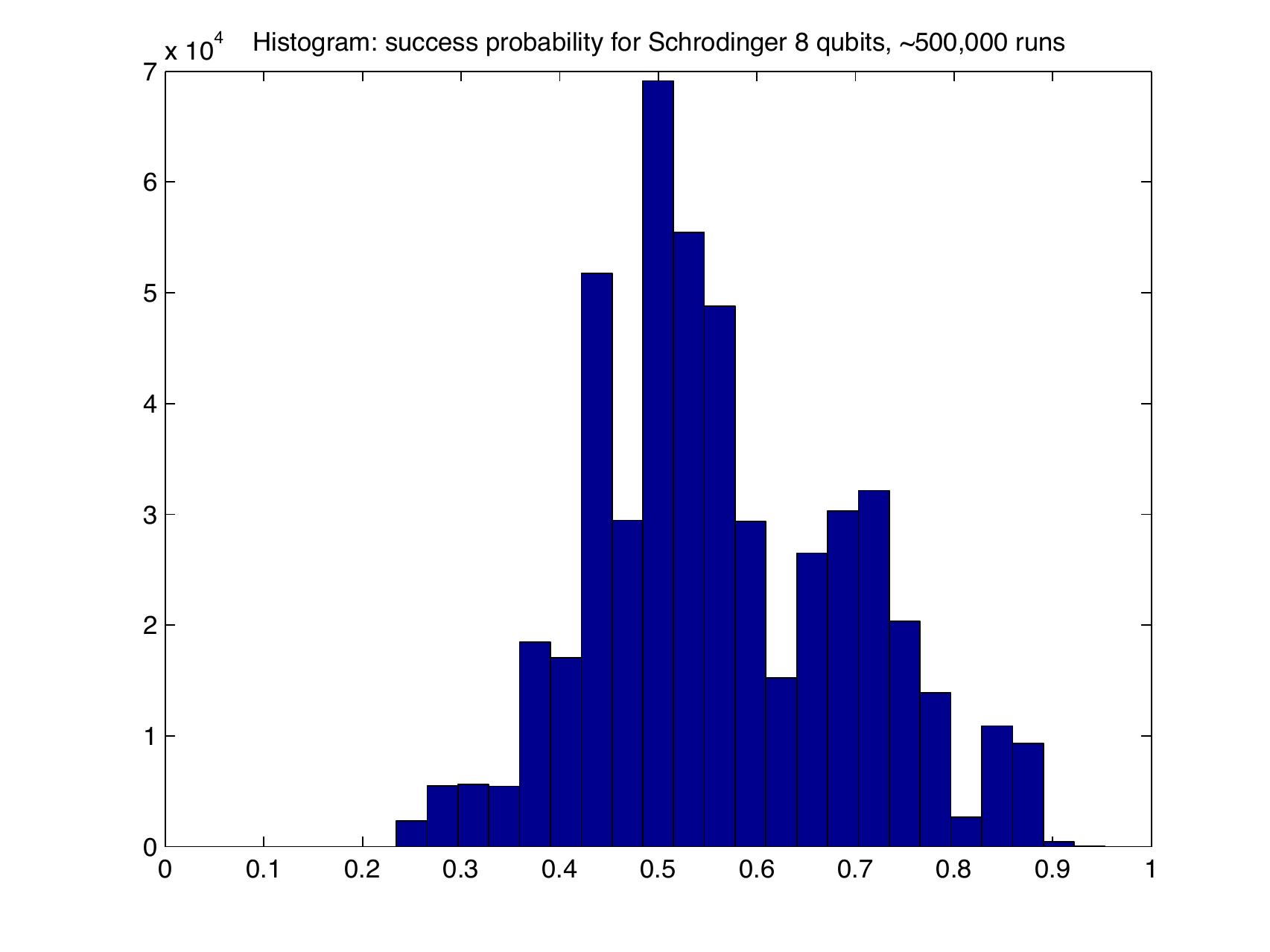}
\emph{Figure 1.} Results obtained via computation of Taylor series solutions of the 8 qubit adiabatic  Schr\"{o}dinger  equation with $T=4$. The histogram displays probability of success outcomes for $500,000$ random Ising Hamiltonians $H_f$ sorted into 32 bins.
\end{minipage}
\hspace{1cm}
\begin{minipage}[t]{0.45\textwidth}
\vspace{1cm}
\includegraphics[width=\textwidth]{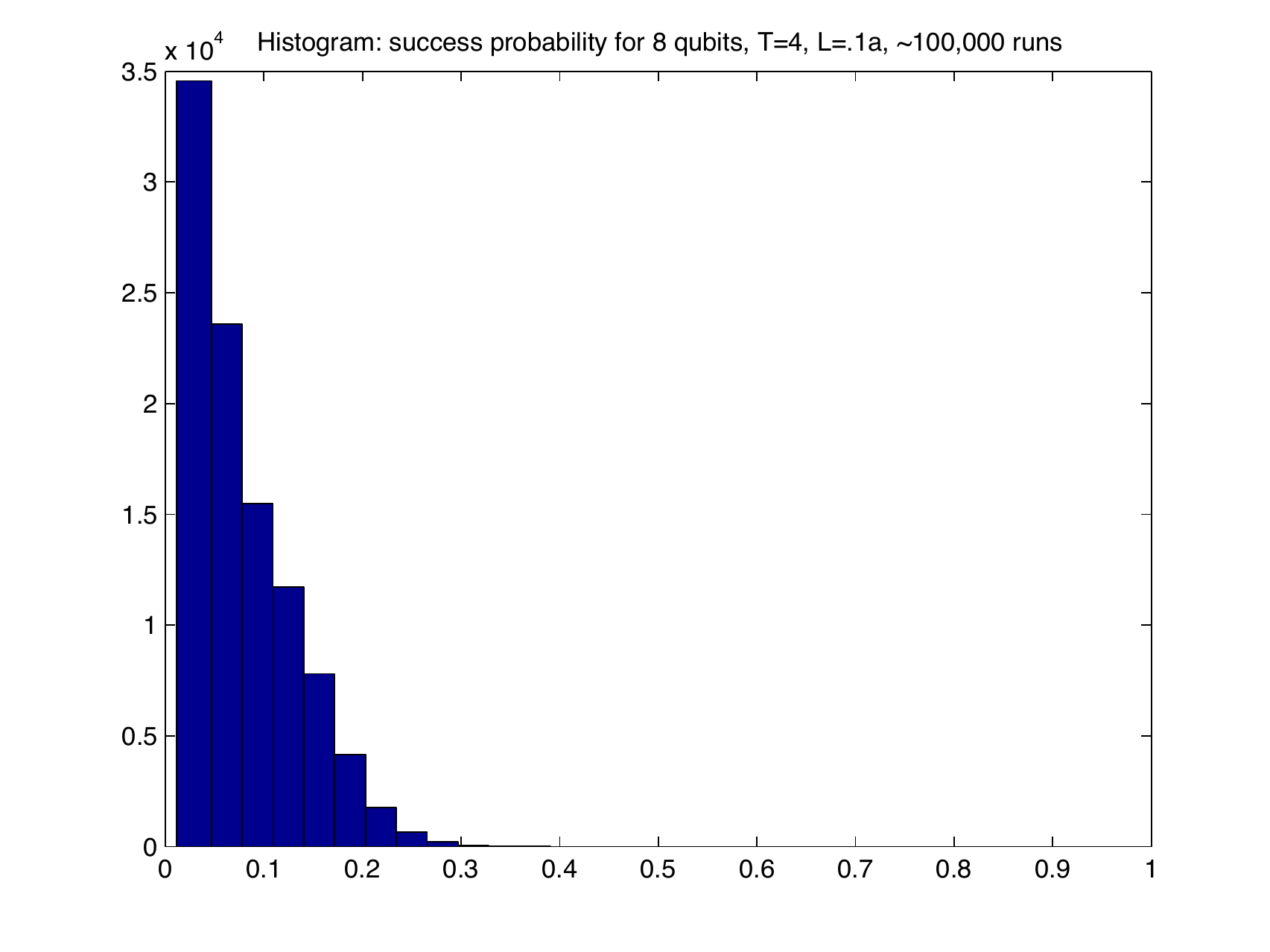}
\emph{Figure 2.} Results obtained via computation of Taylor series solutions of the 8 qubit master equation with the Lindblad operator $L=0.1 \hat{a}$ and $T=4$. The probabilities of success outcomes obtained for $100,000$ random Ising Hamiltonians $H_f$ are sorted into 32 bins.
\end{minipage}

\begin{minipage}[t]{0.45\textwidth}
\vspace{2cm}
\includegraphics[width=\textwidth]{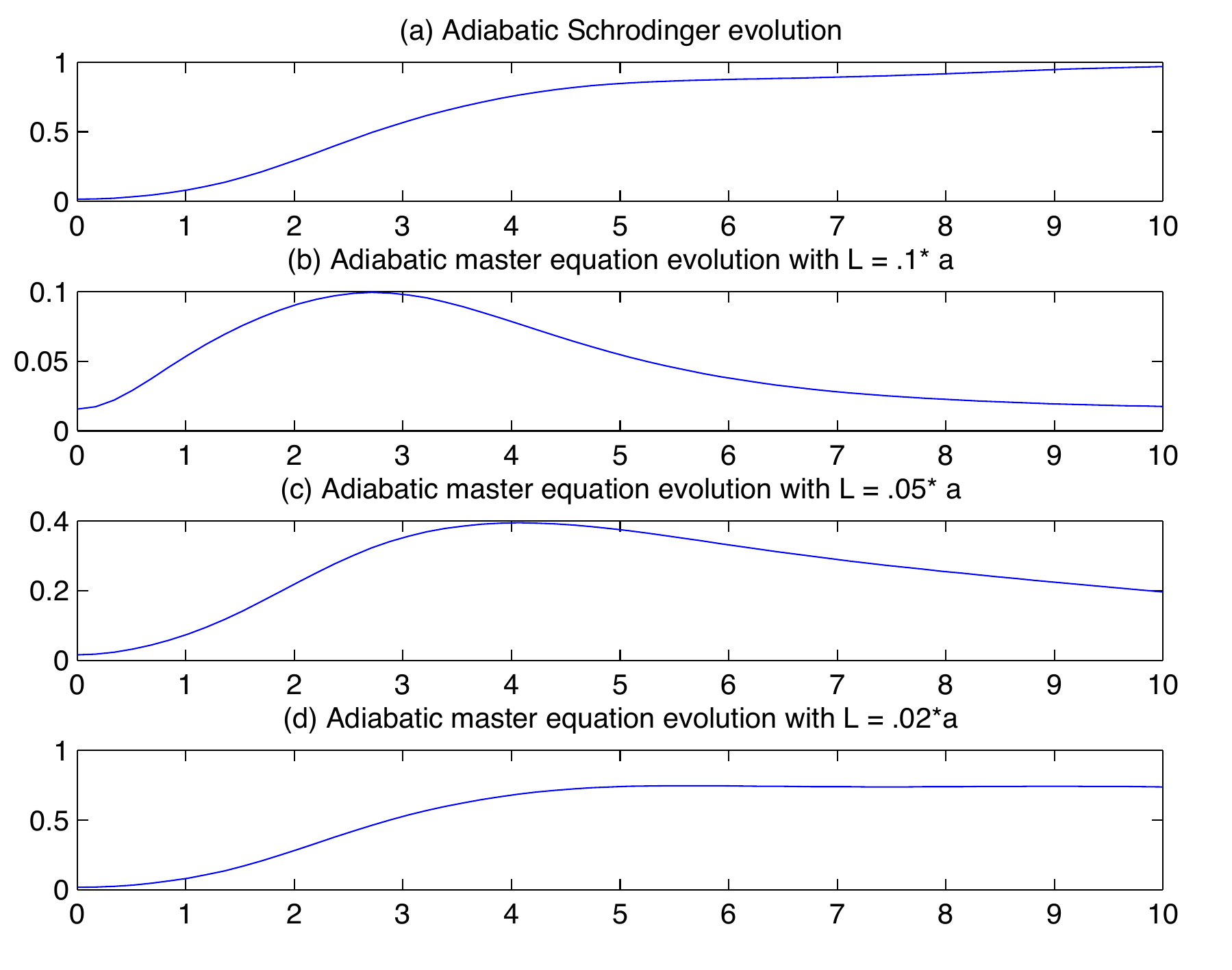}
\emph{Figure 3.} Dependence of the probability of success on $T$ for solutions of the 8 qubit Schr\"{o}dinger equation (a) and the master equations with the Lindblad operators $L=0.1 \hat{a}$ (b), $L=0.05 \hat{a}$ (c), and $L=0.02 \hat{a}$ (d).  Note the inapplicability of the adiabatic theorem in (b) and (c), as well as convergence to the Schr\"{o}dinger solution for diminishing norm of $L$.
 \end{minipage}
\hspace{1cm}
\begin{minipage}[t]{0.45\textwidth}
\vspace{2cm}
\includegraphics[width=\textwidth]{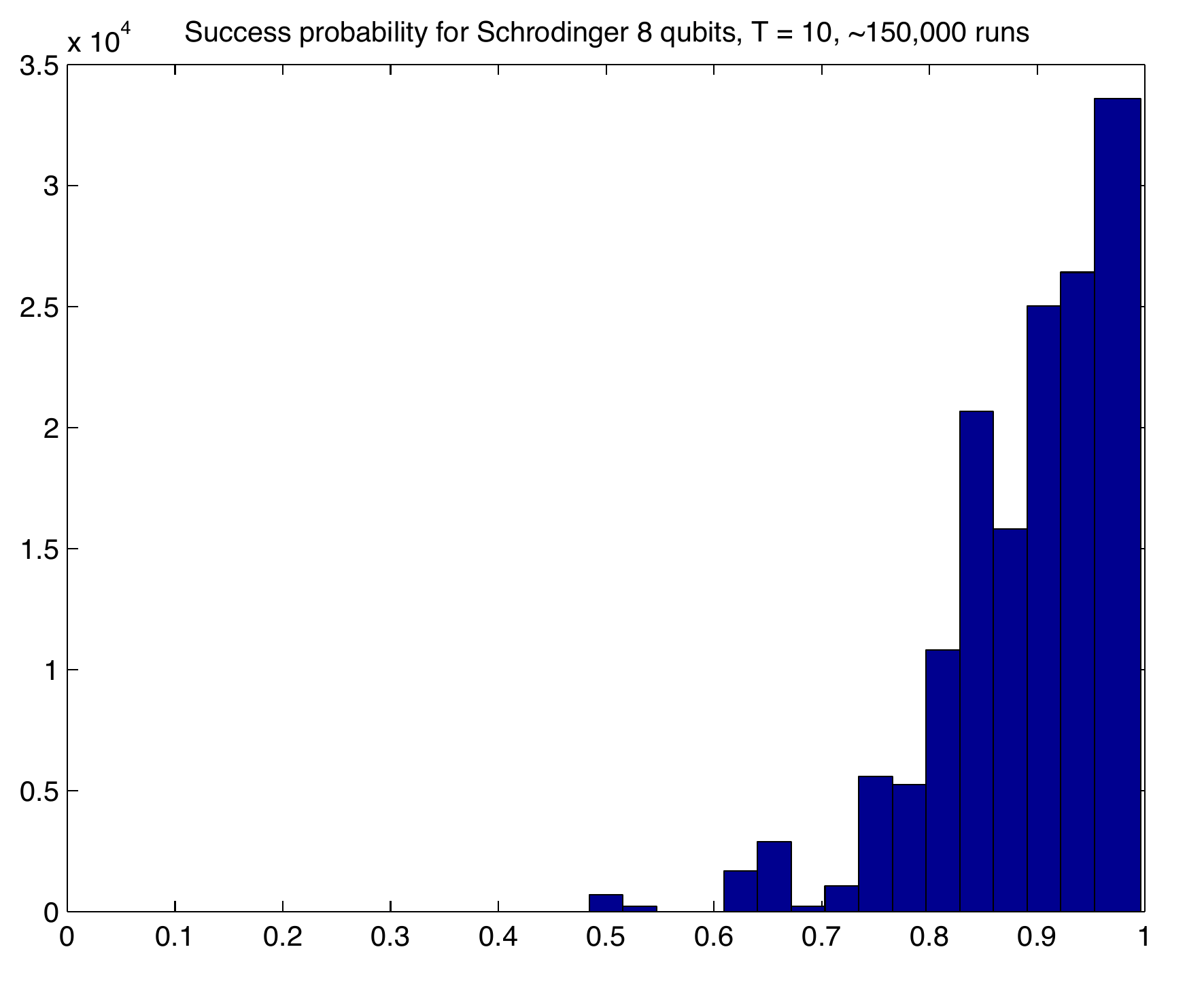}
\emph{Figure 4.} Schr\"odinger evolution as in Fig. 1 with $8$ qubits but with $T=10$. The histogram may be described as multi-modal.
\end{minipage}

\begin{minipage}[t]{0.45\textwidth}
\vspace{1cm}
\includegraphics[width=\textwidth]{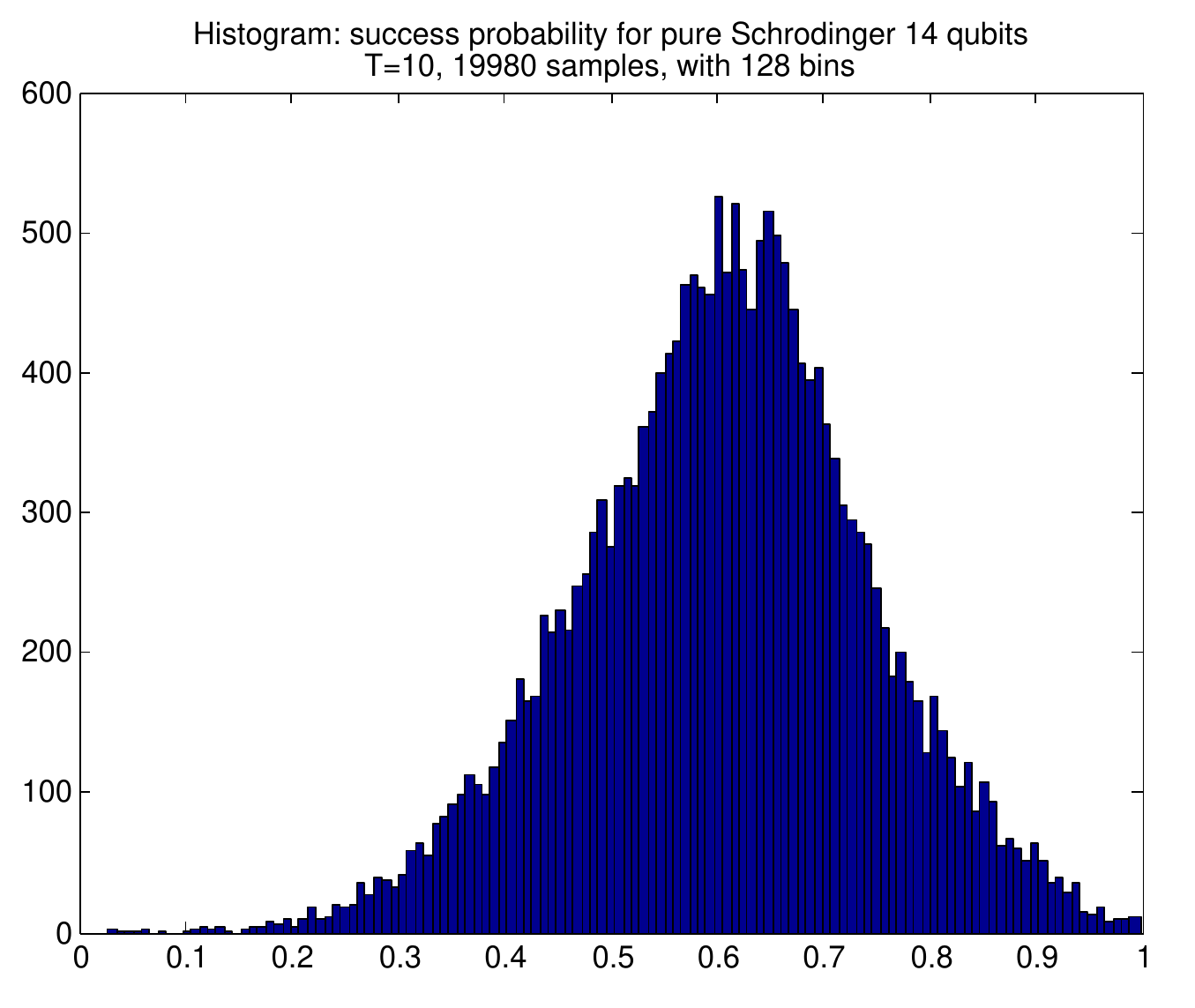}
\emph{Figure 5.} Schr\"odinger evolution with 14 qubits and $T=10$. The computation of these results took about 12 minutes using 90 cores.
\end{minipage}
\hspace{1cm}
\begin{minipage}[t]{0.45\textwidth}
\vspace{1cm}
\includegraphics[width=\textwidth]{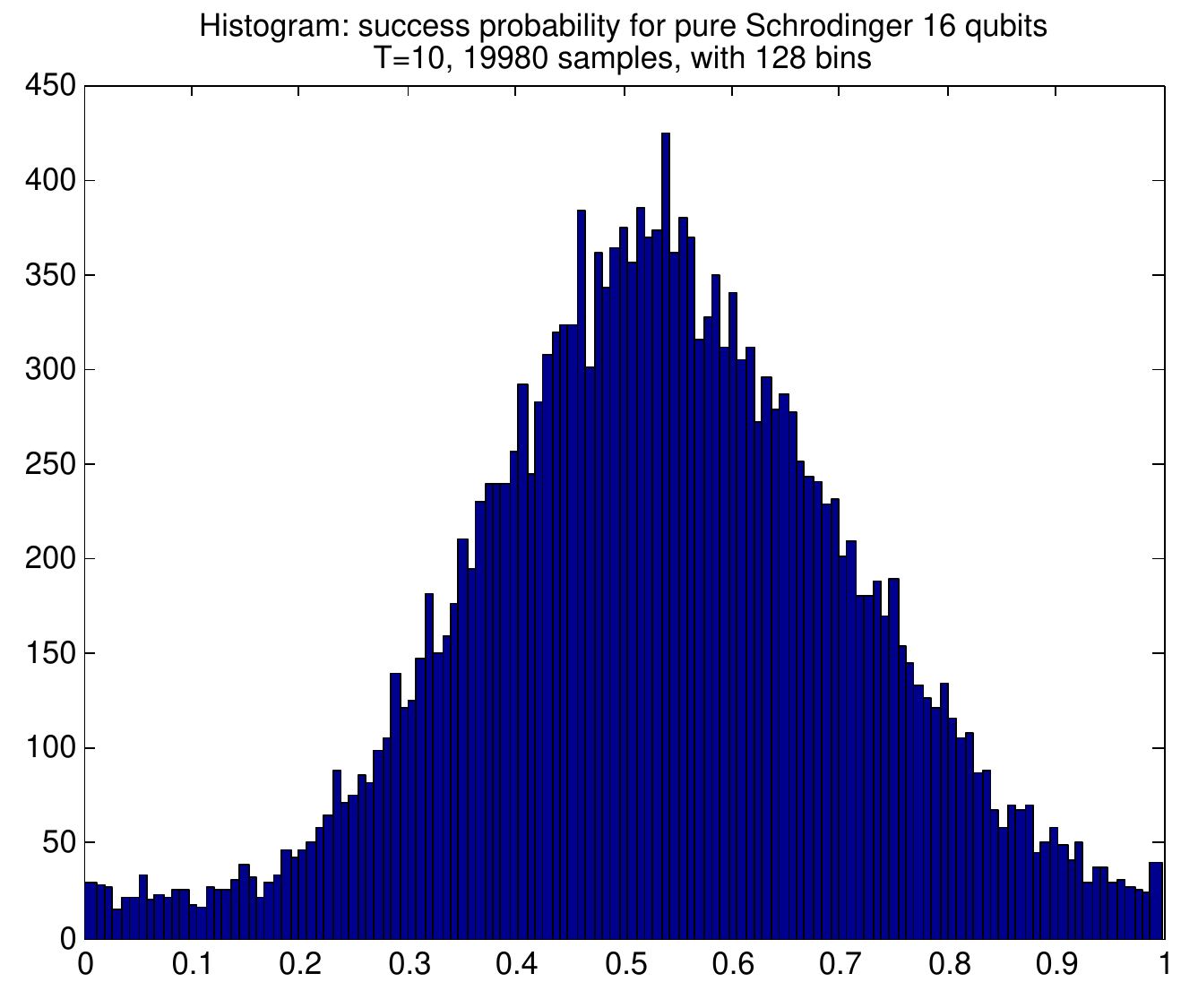}
\emph{Figure 6.} Schr\"odinger evolution with 16 qubits and $T=10$. Computation time with 90 cores was 1 hour and 15 minutes approximately.
\end{minipage}

\begin{minipage}[t]{0.45\textwidth}
\vspace{2cm}
\includegraphics[width=\textwidth]{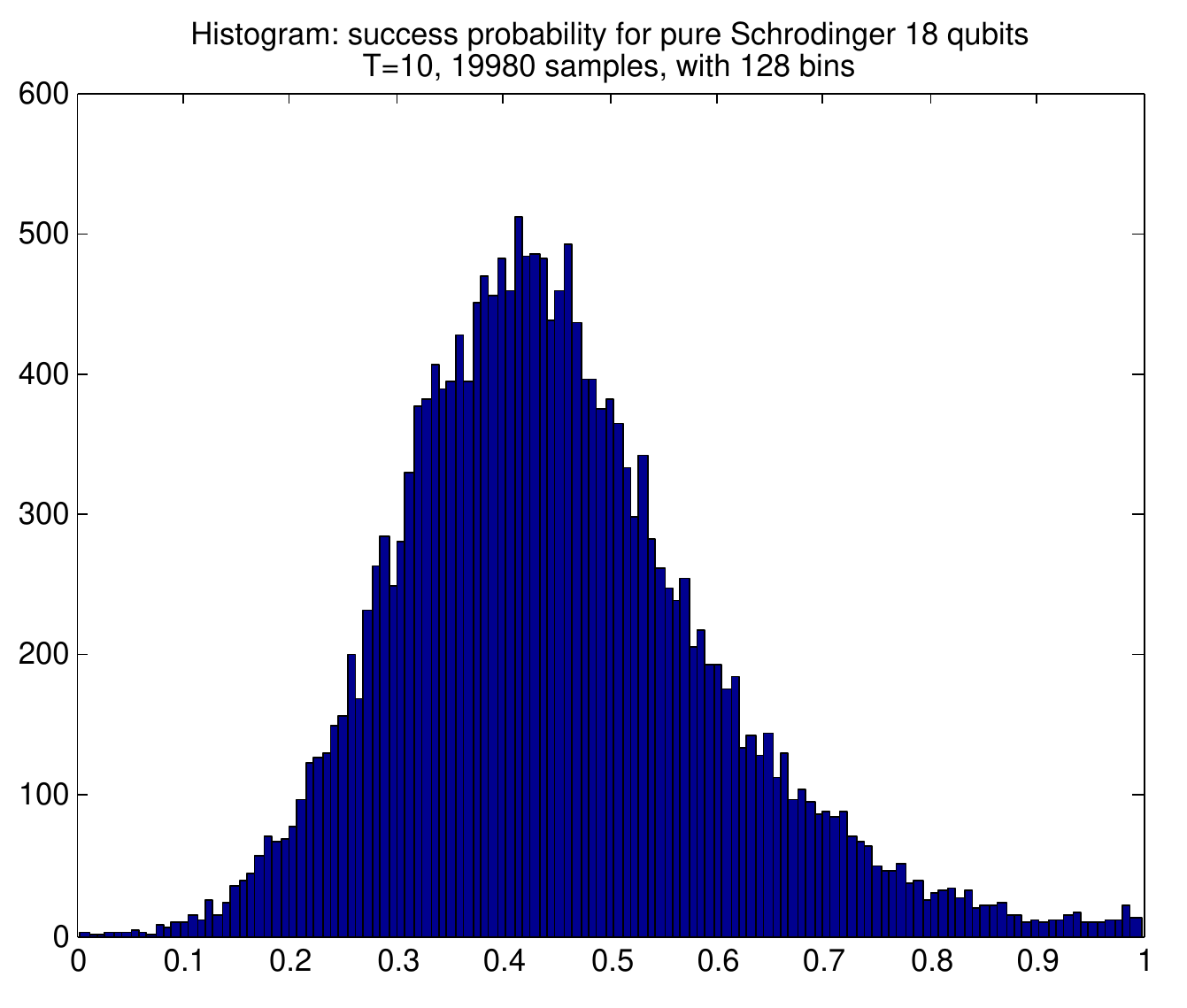}
\emph{Figure 7.} Schr\"odinger evolution with 18 qubits and $T=10$. Computation time was approximately 5 hours and 20 minutes with 90 cores.
\end{minipage}
\hspace{1cm}
\begin{minipage}[t]{0.45\textwidth}
\vspace{2cm}
\includegraphics[width=\textwidth]{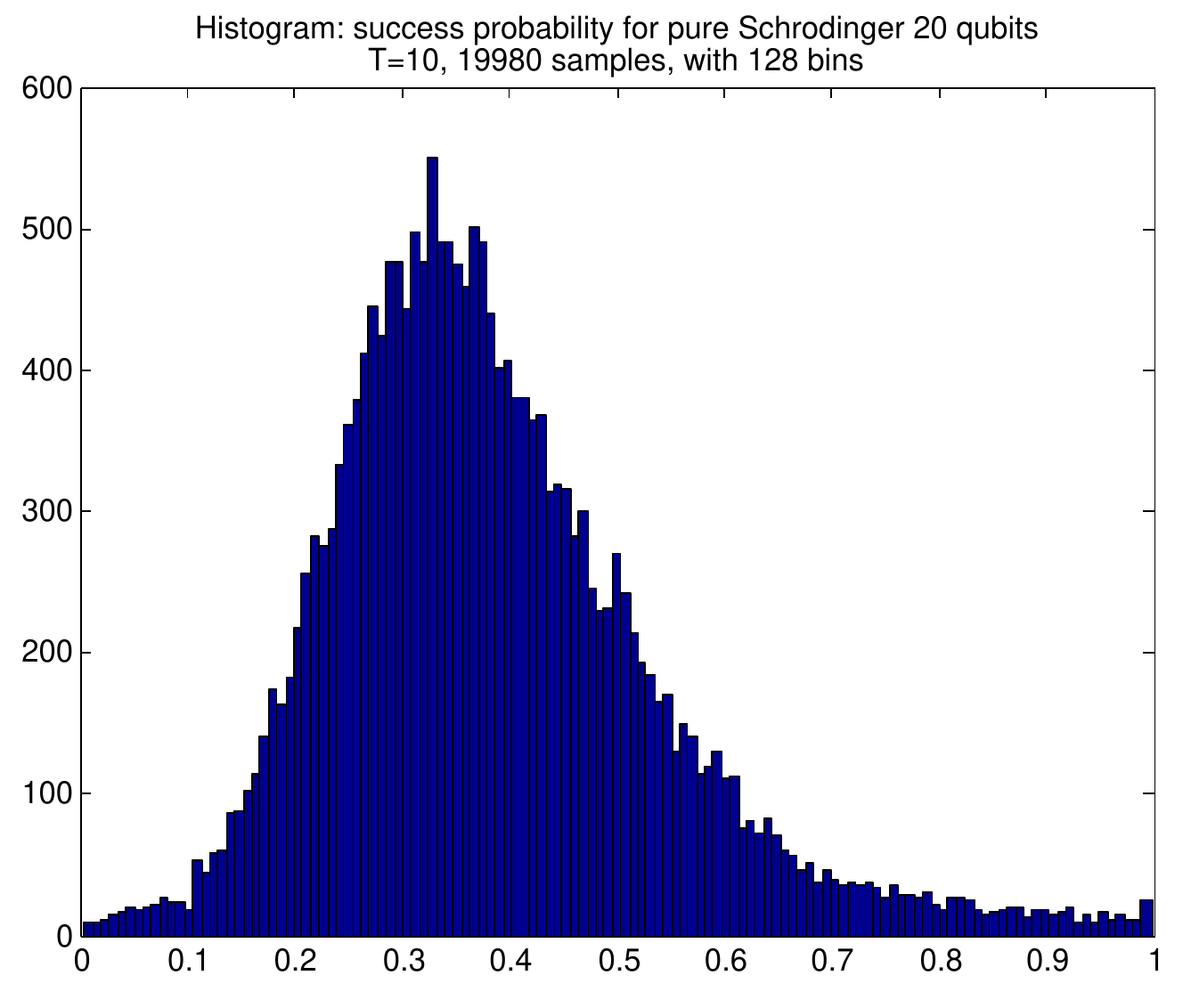}
\emph{Figure 8.} Schr\"odinger evolution with 20 qubits and $T=10$. Using 180 cores, these results were computed in 14 hours.
\end{minipage}

\begin{minipage}[t]{0.45\textwidth}
\vspace{2cm}
\includegraphics[width=\textwidth]{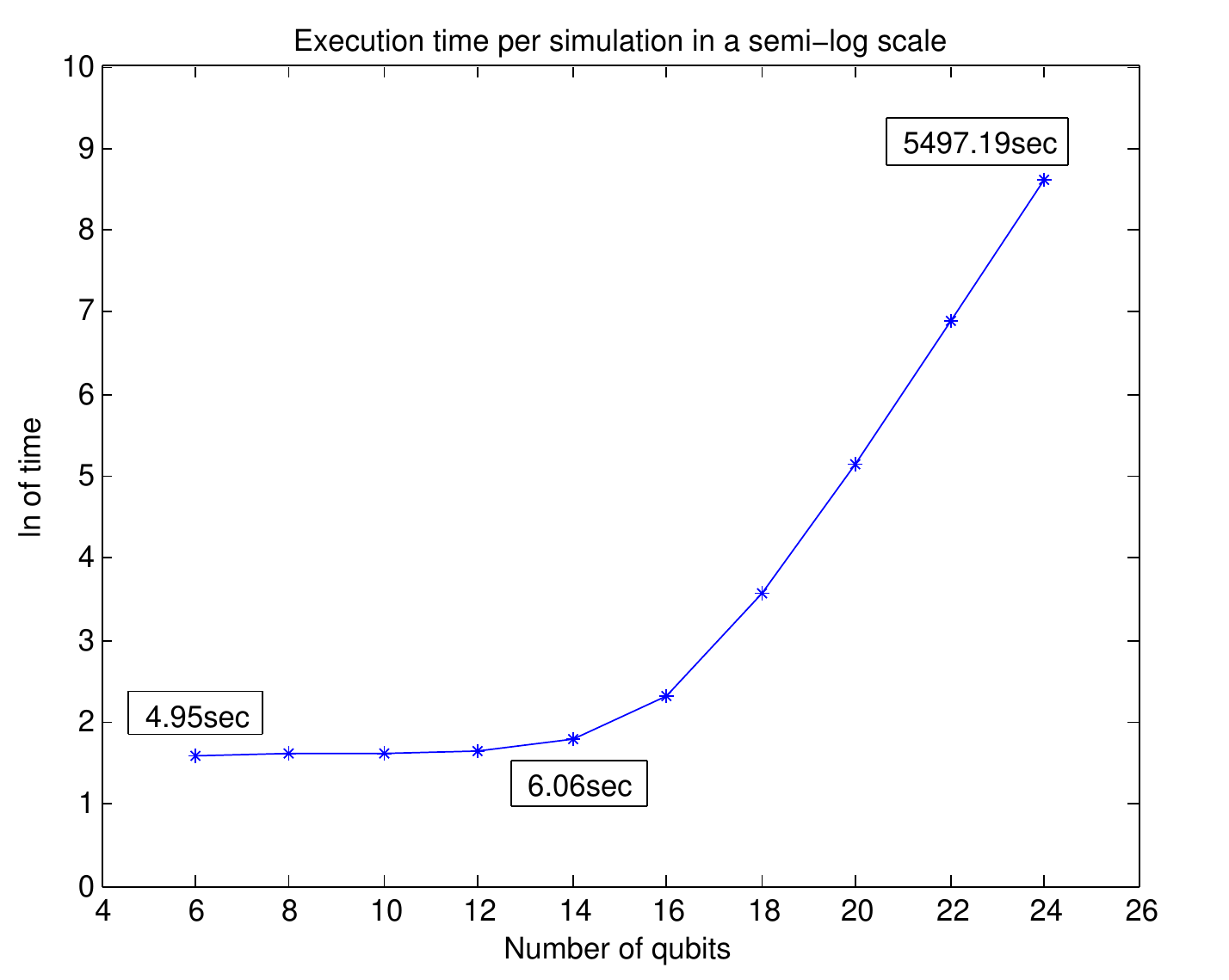}
\emph{Figure 9.} Execution time (wall time) of the algorithm as a function of the number of qubits. Our current implementation running on MATLAB implies a system time larger than the time
required by the algorithm when the number of qubits is too small.
\end{minipage}
\hspace{1cm}
\begin{minipage}[t]{0.45\textwidth}
\vspace{2cm}
\includegraphics[width=\textwidth]{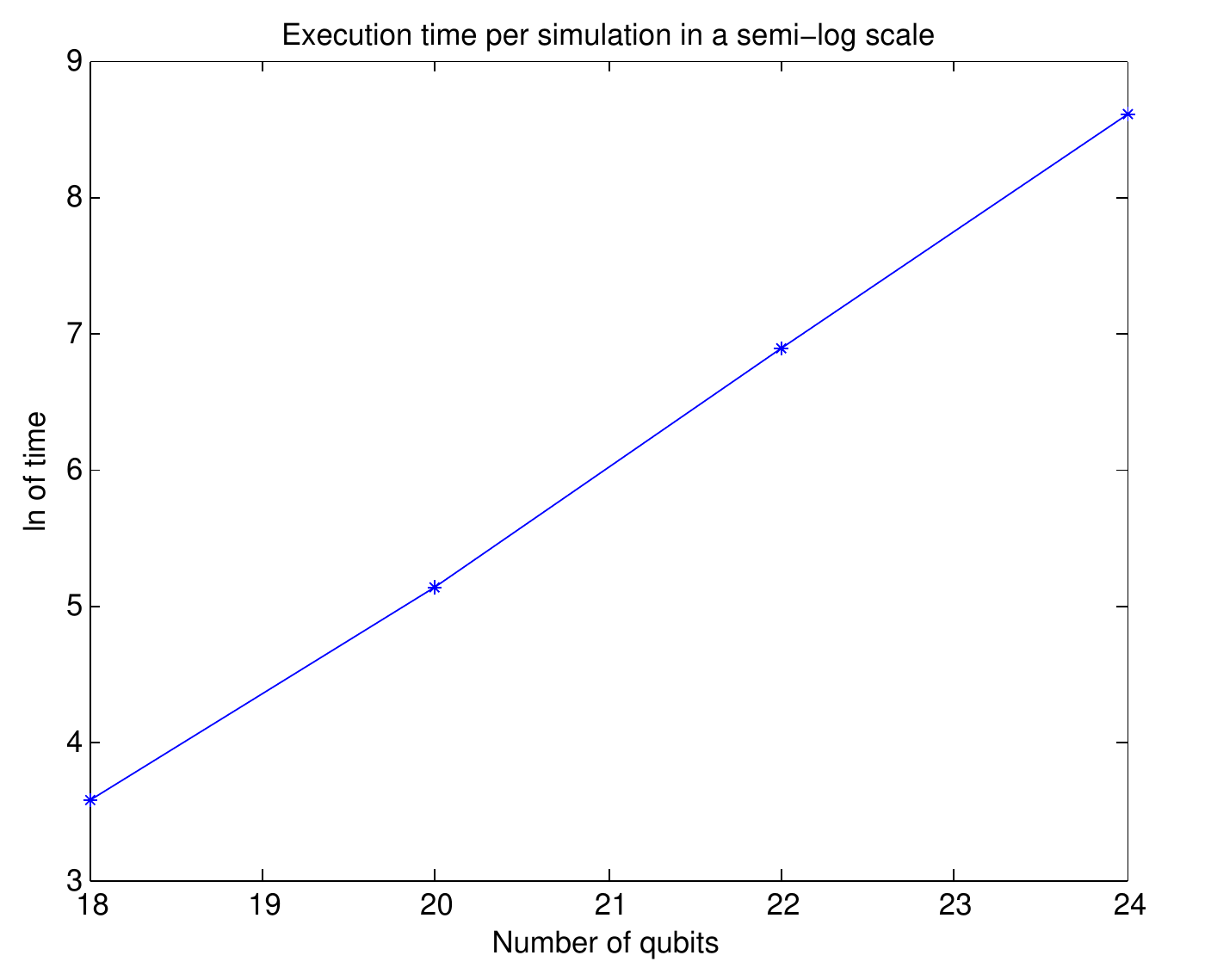}
\emph{Figure 10.} Execution time of the algorithm as a function of the number of qubits. A clear exponential behaviour is observed for larger number of qubits. The exponential function that fits the last
three points (x=20, 22 and 24) is given by: $y=ce^{mx}$ where $c\simeq 6 \mu s $, and $m=0.86$.
\end{minipage}

\begin{minipage}[t]{0.65\textwidth}
\vspace{1cm}
\includegraphics[width=\textwidth]{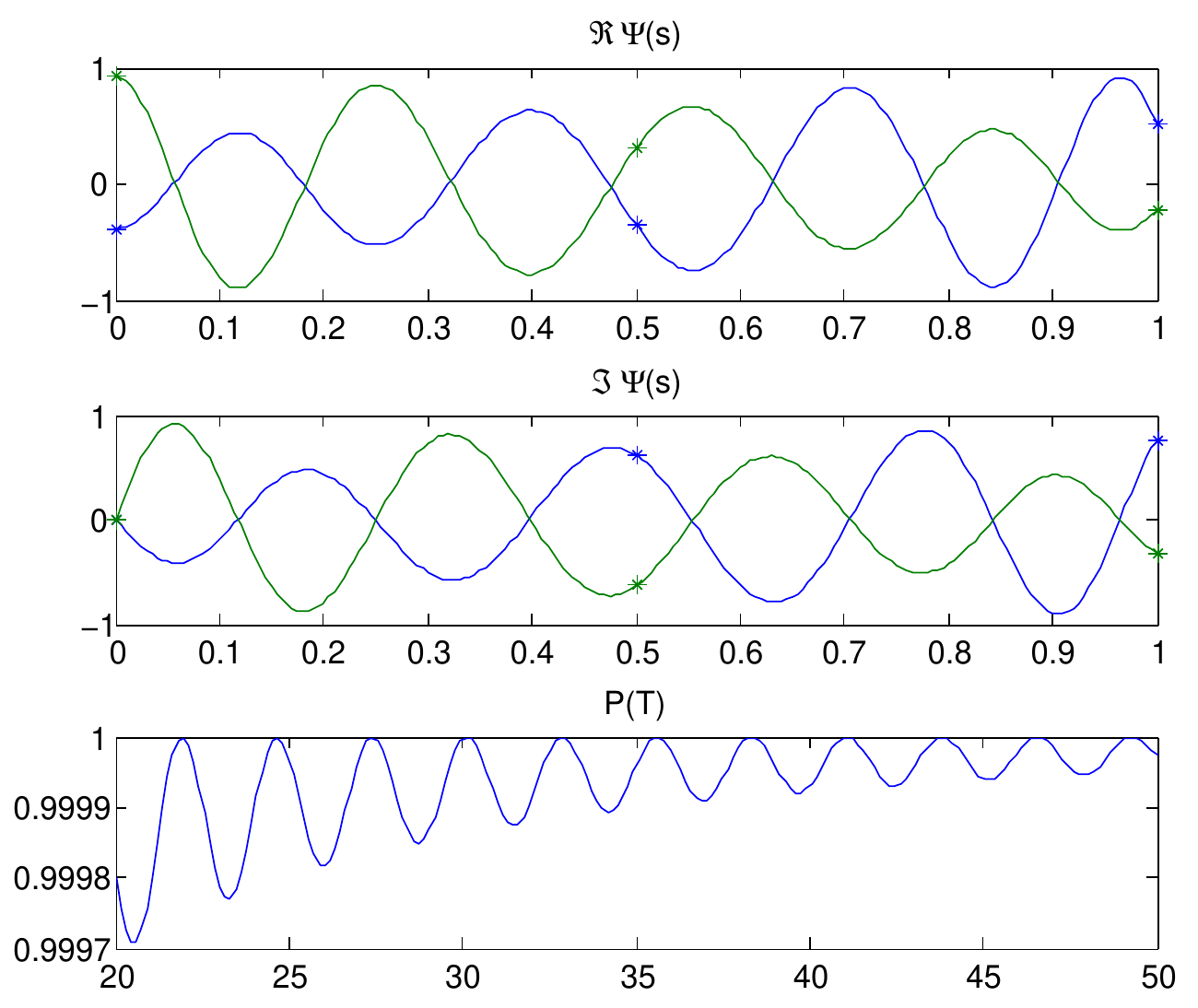}
\emph{Figure 11.} The top two figures show a numerical solution of (\ref{Schr_of_s}) with $T/\hbar = 20$ and the Landau-Zener Hamiltonian $H(s) = (1-2s)\sigma_z + \Delta \sigma_x$ obtained via condensing the Taylor series to $200$ small steps. The star marked points are obtained via computation in just two instalments. The first (resp. second) figure shows the real (resp. imaginary) parts of the two vector components. The bottom graph shows oscillations of the probability of success $P(T)$ as $T/\hbar$ varies between $20$ and $50$.
\end{minipage}

\newpage

\section*{Appendix II: Source codes}

\begin{center}
\framebox [0.9\textwidth]{
\begin{minipage}{0.87\textwidth}
\texttt{
function P=Main(runs,N,T,nint,epsilon)\\
\\
dim=2\^\,N;\\
psi\_0 = (1/sqrt(dim))*ones(dim/2,1);\\
H\_i=Initial\_Hamiltonian(N-1);\\
step=1/nint;\\
for n=1:runs   \\
\phantom{x}\ \ H\_f=Random\_Ising(N);\\
\phantom{x}\ \    ind=(H\_f == min(H\_f));\\
\phantom{x}\ \     psi\_in=psi\_0;\\
\phantom{x}\ \        for k = 1:nint\\
\phantom{x}\ \ \ \            psi\_in = Taylor\_Installements(H\_i,H\_f,T,psi\_in,step,tol,k);\\
\phantom{x}\ \        end\\
\phantom{x}\ \        P[n] = 2*norm(psi\_in(ind))\^\,2;\\
        end
}
\end{minipage}}
\end{center}
Function \texttt{Main}: Function that is being executed on each of the assigned cores on the cluster. Each core will obtains $n=$\texttt{runs} samples. Observe how the Taylor series is calculated for
\texttt{nint} segments in the s-interval. \texttt{N} is the number of qubits, \texttt{T} the final time and \texttt{tol} the tolerance passed to the function \texttt{Taylor\_Installements} below.
\bigskip

\begin{center}
\framebox [0.9\textwidth]{
\begin{minipage}{0.87\textwidth}
\texttt{
function H=Random\_Ising(d)\\
\\
dim=2\^\,d;\\
vals=zeros(1,dim/2);\\
for k=1:d-1\\
\phantom{x}\ \    temp=zeros(1,dim/2\^\,k);\\
\phantom{x}\ \     for r=k+1:d\\
\phantom{x}\ \ \ \     mac=[ones(1,dim/2\^\,r) -ones(1,dim/2\^\,r)];\\
   \phantom{x}\ \ \ \   for p=1:r-(k+1)\\
  \phantom{x}\ \ \ \ \ \        mac=[mac mac];\\
   \phantom{x}\ \ \ \      end\\
       \phantom{x}\ \ \ \  signo=2*randi(2,1)-3;\\
  \phantom{x}\ \ \ \      temp=temp+signo*mac;\\
\phantom{x}\ \    end\\
 \phantom{x}\ \  for p=1:k-1\\
  \phantom{x}\ \ \ \      temp=[temp fliplr(temp)];\\
\phantom{x}\ \     end\\
 \phantom{x}\ \    vals=vals+temp;\\
end\\
H=sparse(vals');
}
\end{minipage}}\\
\end{center}
Function \texttt{Random\_Ising}: Using a fractal-like iteration, creates a sparse vector \texttt{H} of dimension $2^\texttt{d}$ that represents the diagonal of an Ising Hamiltonian. Notice that this vector is palindrome, thus, only the first
$2^\texttt{d-1}$ values are returned.
\bigskip

\begin{center}
\framebox [0.9\textwidth]{
\begin{minipage}{0.87\textwidth}
\texttt{
function H=Initial\_Hamiltonian(d)\\
\\
dim=2\^\,d;\\
rows=[1 2];\\
cols=[2 1];\\
for k=1:d-1\\
\phantom{x}\ \    rows=[rows rows+2\^\,k];\\
\phantom{x}\ \    cols=[cols cols+2\^\,k];\\
\phantom{x}\ \    p=1:2\^\,k;\\
\phantom{x}\ \    rows=[rows p p+2\^\,k];\\
\phantom{x}\ \    cols=[cols p+2\^\,k p];\\
end\\
vals=-ones(1,size(rows,2));\\
H = sparse(rows, cols, vals, dim, dim);
}
\end{minipage}}\\
\end{center}
Function \texttt{Initial\_Hamiltonian}: Using a fractal-like iteration, creates a sparse matrix \texttt{H} (the initial Hamiltonian) of dimension $2^\texttt{d}$. Notice that for \texttt{N} qubits we use \texttt{Initial\_Hamiltonian(N-1)} due to the
symmetry in \texttt{H\_f}.
\bigskip

\begin{center}
\framebox [0.9\textwidth]{
\begin{minipage}{0.87\textwidth}
\texttt{
function psi=Taylor\_Installements(H\_i,H\_f\_diag,T,psi\_in,step,tol,k)\\
\\
\%******* Prepares Auxiliary Operators A and B of Eq. (\ref{auxil}) for the recurrence \\
c = -1i * T;\\
d = (k - 1) * step;\\
i\_by\_psi\_in = H\_i * psi\_in - flipud(psi\_in);\\
f\_by\_psi\_in = H\_f\_diag .* psi\_in;\\
psi\_n\_min\_2 = psi\_in;\\
psi\_n\_min\_1 = c * ((1-d) * i\_by\_psi\_in + d * f\_by\_psi\_in);\\
psi = psi\_n\_min\_1 * step + psi\_n\_min\_2;\\
nrm\_cor = 1;\\
n = 1;\\
i\_by\_min\_2 = i\_by\_psi\_in;\\
f\_by\_min\_2 = f\_by\_psi\_in;\\
\\
\%******** Implements the recurrence as in Eq. (\ref{recurrence}) \\
while (nrm\_cor > tol)\\
\phantom{x}\ \     n = n+1;\\
\phantom{x}\ \     i\_by\_min\_1 = H\_i * psi\_n\_min\_1 - flipud(psi\_n\_min\_1);\\
 \phantom{x}\ \    f\_by\_min\_1 = H\_f\_diag .* psi\_n\_min\_1;\\
 \phantom{x}\ \    psi\_n = (c/n)*((1-d) * i\_by\_min\_1 + d * f\_by\_min\_1 + f\_by\_min\_2 - i\_by\_min\_2);\\
 \phantom{x}\ \    cor = psi\_n*step\^\,n;\\
 \phantom{x}\ \    nrm\_cor = norm(cor);\\
\phantom{x}\ \     psi = psi + cor;\\
 \phantom{x}\ \    psi\_n\_min\_2 = psi\_n\_min\_1;\\
 \phantom{x}\ \    psi\_n\_min\_1 = psi\_n;\\
 \phantom{x}\ \    i\_by\_min\_2 = i\_by\_min\_1;\\
 \phantom{x}\ \    f\_by\_min\_2 = f\_by\_min\_1;\\
end
}
\end{minipage}}\\
\end{center}
Function \texttt{Taylor\_Installements}: Favouring vector-vector calculations, this function computes the Taylor series for the \texttt{k}-th segment of the s-interval. Notice that the calculation is stopped when the contribution at the next step
is smaller than the tolerance \texttt{tol}.

\end{document}